\documentclass[10pt]{llncs}

\usepackage{amsmath,amssymb}
\usepackage{ntheorem}
\usepackage{subfigure}
\usepackage{tikz}
\usepackage{times}
\usepackage{algorithm}
\usepackage[margin=3.4cm]{geometry}

\newtheorem{prop}{Proposition}
\newtheorem{lemme}{Lemma}
\newtheorem{cond}{Condition}
\newtheorem{hypo}{Progression Hypothesis}

\newcommand{\carre}{\hfill$\square$}
\newcommand{\G}{\ensuremath {{\cal G}}}
\newcommand{\T}{\ensuremath {\mathbb{T}}}
\newcommand{\ST}{\ensuremath {{\cal S}}_{\T}}
\newcommand{\SG}{\ensuremath {{\cal S}}_G}
\newcommand{\EG}{\ensuremath \G=(G, \SG, \ST)}
\newcommand{\J}{\ensuremath {{\cal J}}}
\newcommand{\leadstoo}{\overset{st}{\leadsto}}

\usepackage{url}

\begin{document}
\title{\large Characterizing Topological Assumptions of Distributed Algorithms\\ in Dynamic Networks\thanks{A preliminary version of this paper appeared in SIROCCO~\cite{CCF09}.}}
\author{Arnaud Casteigts$^1$ \and Serge Chaumette$^2$ \and Afonso Ferreira$^3$}

\institute{SITE, University of Ottawa, Canada.\\\email{casteig@site.uottawa.ca}\medskip \and LaBRI, Universit\'e de Bordeaux, France.\\\email{serge.chaumette@labri.fr}\medskip \and INRIA-MASCOTTE, Sophia Antipolis, France.\\\email{afonso.ferreira@cnrs-dir.fr}}
\maketitle

\begin{abstract} 
  \small
  Besides the complexity in time or in number of messages, a common approach for analyzing distributed algorithms is to look at the assumptions they make on the underlying network. We investigate this question from the perspective of network dynamics. In particular, we ask how a given property on the evolution of the network can be rigorously proven as necessary or sufficient for a given algorithm. The main contribution of this paper is to propose the combination of two existing tools in this direction: {\em local computations} by means of {\em graph relabelings}, and {\em evolving graphs}. Such a combination makes it possible to express fine-grained properties on the network dynamics, then examine what impact those properties have on the execution at a precise, intertwined, level. We illustrate the use of this framework through the analysis of three simple algorithms, then discuss general implications of this work, which include (i)~the possibility to compare distributed algorithms on the basis of their topological requirements, (ii)~a formal hierarchy of dynamic networks based on these requirements, and (iii)~the potential for mechanization induced by our framework, which we believe opens a door towards automated analysis and decision support in dynamic networks.
\end{abstract}

\section{Introduction}
\label{sec:intro}

The past decade has seen a burst of research in the field of communication networks. This is particularly true for dynamic networks due to the arrival, or impending deployment, of a multitude of applications involving new types of communicating entities such as {\em wireless sensors}, {\em smartphones}, {\em satellites}, {\em vehicles}, or swarms of {\em mobile robots}. These contexts offer both unprecedented opportunities and challenges for the research community, which is striving to design appropriate algorithms and protocols. Behind the apparent unity of these networks lies a great diversity of assumptions on their dynamics. One end of the spectrum corresponds to {\em infrastructured} networks, in which only terminal nodes are dynamic --~these include 3G/4G telecommunication networks, access-point-based {\em Wi-Fi} networks, and to some extent the Internet itself. At the other end lies {\em delay-tolerant networks} (DTNs), which are characterized by the possible absence of end-to-end communication route at any instant. The defining property of DTNs actually reflects many types of real-world contexts, from satellites or vehicular networks to pedestrian or social animal networks (e.g. birds, ants, termites). In-between lies a number of environments whose capabilities and limitations require specific attention.

A consequence of this diversity is that a given protocol for dynamic networks may prove appropriate in one context, while performing poorly (or not at all) in another. The most common approach for evaluating protocols in dynamic networks is to run simulations, and use a given {\em mobility model} (or set of traces) to generate topological changes during the execution. These parameters must faithfully reflect the target context to yield an accurate evaluation. Likewise, the comparison between two protocols is only meaningful if similar traces or mobility models are used. This state of facts makes it often ambiguous and difficult to judge of the appropriateness of solutions based on the sole experimental results reported in the literature. The problem is even more complex if we consider the possible biases induced by further parameters like the size of the network, the density of nodes, the choice of PHY or MAC layers, bandwidth limitations, latency, buffer size, {\it etc.}

The fundamental requirement of an algorithm on the network dynamics will likely be better understood from an {\em analytical} standpoint, and some recent efforts have been carried out in this direction. They include the works by O'Dell {\it et al.}~\cite{OW05} and Kuhn {\it et al.}~\cite{KLO10}, in which the impacts of given assumptions on the network dynamics are studied for some basic problems of distributed computing ({\em broadcast}, {\em counting}, and {\em election}). These works have in common an effort to make the dynamics amenable to analysis through exploiting properties of a {\em static} essence: even though the network is possibly highly-dynamic, it remains {\em connected} at every instant. The approach of population protocols~\cite{AngluinADFP2006,angluin2007} also contributed to more analytical understanding. Here, no assumptions are made on the network connectivity {\em at a given instant}, but yet, the same fundamental idea of looking at dynamic networks through the eyes of static properties is leveraged by the concept of {\em graph of interaction}, in which every entity is assumed to interact infinitely often with its neighbors (and thus, dynamics is reduced to a scheduling problem in static networks). Besides the fact that the above assumptions are strong --~we will show how strong in comparison to others in a hierarchy~--, we believe that the very attempt to {\em flatten} the time dimension does prevent from understanding the true requirements of an algorithm on the network dynamics. 

As a trivial example, consider the broadcasting of a piece of information in the network depicted in Figure~\ref{fig:abc}. The possibility to complete the broadcast in this scenario clearly depends on which node is the initial emitter: $a$ and $b$ may succeed, while $c$ cannot. Why? How can we express this intuitive property the topology evolution must have with respect to the emitter and the other nodes? Flattening the time-dimension without keeping information on the ordering of events would obviously loose some important specificities, such as the fact that nodes $a$ and $c$ are in a non-symmetrical configuration. How can we prove, more generally, that a given assumption on the dynamics is necessary or sufficient for a given problem (or algorithm)? How can we find (and define) property that relate to finer-grain aspects than recurrence or more generally {\em regularities}. Even when intuitive, {\em rigorous} characterizations of this kind might be difficult to obtain without appropriate models and formalisms --~a conceptual shift is needed.

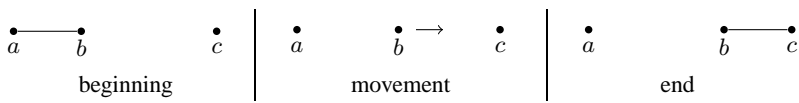
\begin{figure}[h]
  \begin{center}
    \begin{tabular}{c|c|c}
      \scriptsize
      \begin{tikzpicture}[scale=.9]
        \tikzstyle{every node}=[fill=black, circle, inner sep=1pt]
        \path (0,0) node (a) {};
        \path (1,0) node (b) {};
        \path (3,0) node (c) {};
        \path (0,.3) coordinate (tanchor) {};
        \tikzstyle{every node}=[font=\footnotesize]
        \path (a)+(0,-.25) node (la) {$a$};
        \path (b)+(0,-.25) node (lb) {$b$};
        \path (c)+(0,-.25) node (lc) {$c$};
        \draw (a)--(b);
      \end{tikzpicture}
      \quad\quad&\quad
      \begin{tikzpicture}[scale=.9]
        \tikzstyle{every node}=[fill=black, circle, inner sep=1pt]
        \path (0,0) node (a) {};
        \path (1.5,0) node (b) {};
        \path (3,0) node (c) {};
        \path (0,.3) coordinate (tanchor) {};
        \tikzstyle{every node}=[font=\footnotesize]
        \path (a)+(0,-.25) node (la) {$a$};
        \path (b)+(0,-.25) node (lb) {$b$};
        \path (c)+(0,-.25) node (lc) {$c$};
        
        \draw[->, shorten >=20pt, shorten <=5pt] (b)--(c);
      \end{tikzpicture}
      \quad\quad&\quad
      \begin{tikzpicture}[scale=.9]
        \tikzstyle{every node}=[fill=black, circle, inner sep=1pt]
        \path (0,0) node (a) {};
        \path (2,0) node (b) {};
        \path (3,0) node (c) {};
        \path (0,.3) coordinate (tanchor) {};
        \tikzstyle{every node}=[font=\footnotesize]
        \path (a)+(0,-.25) node (la) {$a$};
        \path (b)+(0,-.25) node (lb) {$b$};
        \path (c)+(0,-.25) node (lc) {$c$};
        \draw (b)--(c);
      \end{tikzpicture}
      \\
      beginning&movement&end
    \end{tabular}
  \end{center}
  \caption{\label{fig:abc}A basic scenario, where a node ($b$) moves during the execution.}
\end{figure}

We investigate these questions in the present paper. Contrary to the aforementioned approaches, in which a given context is first considered, then the feasibility of problems studied in this particular context, we suggest the somehow reverse approach of considering first a problem, then trying to characterize its {\em necessary} and/or {\em sufficient} conditions (if any) in terms of network dynamics. We introduce a general-purpose analysis framework based on the combination of 1) {\em local computations} by means of {\em graph relabelings}~\cite{ac:litovsky99}, and 2) an appropriate formalism for dynamic networks, {\em evolving graphs}~\cite{Fer04}, which formalizes the evolution of the network topology as an ordered sequence of static graphs. The strengths of this combination are several: First, the use of local computations allows to obtain general impossibility results that do not depend on a particular communication model ({\it e.g.,} {\em message passing}, {\em mailbox}, or {\em shared memory}). Second, the use of evolving graphs enables to express {\em fine-grain} network properties that remain temporal in essence. (For instance, a necessary condition for the broadcast problem above is the existence of a temporal path, or {\em journey}, from the emitter to any other node, which statement can be expressed using monadic second-order logic on evolving graphs.) 
The combination of graph relabelings and evolving graphs makes it possible to study the execution of an algorithm as an intertwined sequence of topological events and computations, leading to a precise characterization of their relation. 
The framework we propose should be considered as a {\em conceptual framework} to guide the analysis of distributed algorithms. As such, it is specified at a high-level of abstraction and does not impose the choice for, say, a particular logic ({\it e.g.} first-order {\it vs.} LMSO) or scope of computation ({\it e.g.} pairwise {\it vs.} starwise interaction), although all our examples assume LMSO and pairwise interactions. Finally, we believe this framework could pave the way to decision support systems or mechanized analysis in dynamic networks, both of which are discussed as possible applications.

Local computations and evolving graphs are first presented in Section~\ref{sec:existing}, together with central properties of dynamic networks (such as {\em connectivity over time}, whose intuitive implications on the broadcast problem were explored in various work --~see {\it e.g.} \cite{awerbuch84,BFJ03}). We describe the analysis framework based on the combination of both tools in Section~\ref{sec:framework}. This includes the reformulation of an execution in terms of {\em relabelings over a sequence of graphs}, as well as new formulations of what a necessary or sufficient condition is in terms of existence and non-existence of such a relabeling sequence. We illustrate these theoretical tools in Section~\ref{sec:analysis} through the analysis of three basic examples, {\it i.e.,} one broadcast algorithm and two counting algorithms, one of which can also be used for election. (Note that our framework was recently applied to the problem of {\em mutual exclusion} in~\cite{FGA11}.) The rest of the paper is devoted to exploring some implications of the proposed approach, articulated around the two major motifs of {\em classification} (Section~\ref{sec:classification}) and {\em mechanization} (Section~\ref{sec:mechanization}). The section on classification discusses how the conditions resulting from analysis translate into more general properties that define classes of evolving graphs. The relations of inclusion between these classes are examined, and interestingly-enough, they allow to organize the classes as a {\em connected} hierarchy. We show how this classification can reciprocally be used to evaluate and compare algorithms on the basis of their topological requirements. The section on mechanization discusses to what extent the tasks related to assessing the appropriateness of an algorithm in a given context can be automated. We provide canonical ways of checking inclusion of a given network trace in all classes resulting from the analyses in this paper (in efficient time), and mention some ongoing work around the use of the {\em coq} proof assistant in the context of local computation, which we believe could be extended to evolving graphs. Section~\ref{sec:conclusion} eventually concludes with some remarks and open problems.

\section{Related work -- the building blocks}
\label{sec:existing}
This section describes the building blocks of the proposed analysis framework, that are, \textit{Local Computations} to abstract the communication model, \textit{Graph Relabeling Systems} as a formalism to describe local computations, and \textit{Evolving Graphs} to express fine-grained properties on the network dynamics. Reading this section is required for a clear understanding of the subsequent ones.

\subsection{Abstracting communications through local computations and graph relabelings}
\label{sec:grs}

Distributed algorithms can be expressed using a variety of communication models (\textit{e.g.} message passing, mailboxes, shared memory). Although a vast majority of algorithms is designed in one of these models --~predominantly the message passing model~--, the very fact that one of them is chosen implies that the obtained results (\textit{e.g.} positive or negative characterizations and associated proofs) are limited to the scope of this model. This problem of diversity among formalisms and results, already pointed out twenty years ago in~\cite{lynch89hundred}, led researchers to consider higher abstractions when studying fundamental properties of distributed systems. 

\emph{Local computations} and \emph{Graph relabelings} were jointly proposed in this perspective in~\cite{ac:litovsky99}. These theoretical tools allow to represent a distributed algorithm as a set of local interaction rules that are independent from the effective communications. Within the formalism of graph relabelings, the network is represented by a graph whose vertices and edges are associated with labels that represent the algorithmic state of the corresponding nodes and links. An interaction rule is then defined as a transition pattern $(preconditions, actions)$, where $preconditions$ and $actions$ relate to these labels values. Since the interactions are local, each transition pattern must involve a limited and connected subset of vertices and edges. Figure~\ref{fig:local_models} shows different scopes of computation, which are not necessarily the same for $preconditions$ and $actions$. 


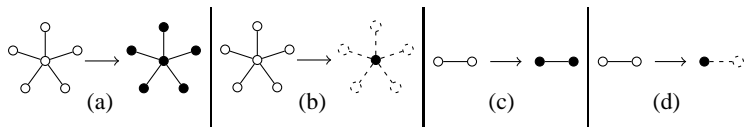
\begin{figure}[h]
  \begin{center}
    \begin{tabular}{c|c|c|c}
    \subfigure[]{\label{fig:most}
    \begin{tikzpicture}[xscale=.45,yscale=.45]
      \tikzstyle{every node}=[draw, circle, inner sep=1.2pt]
      \path (1,1) node (a) {};
      \path (a)+(18:1) node (b) {};
      \path (a)+(90:1) node (c) {};
      \path (a)+(162:1) node (d) {};
      \path (a)+(234:1) node (e) {};
      \path (a)+(306:1) node (f) {};
      \draw (a)--(b);
      \draw (a)--(c);
      \draw (a)--(d);
      \draw (a)--(e);
      \draw (a)--(f);
      \tikzstyle{every node}=[draw, fill, circle, inner sep=1.2pt]
      \path (4.5,1) node (a) {};
      \path (a)+(18:1) node (b) {};
      \path (a)+(90:1) node (c) {};
      \path (a)+(162:1) node (d) {};
      \path (a)+(234:1) node (e) {};
      \path (a)+(306:1) node (f) {};
      \draw (a)--(b);
      \draw (a)--(c);
      \draw (a)--(d);
      \draw (a)--(e);
      \draw (a)--(f);
      \draw[->, shorten >= 15pt, shorten <= 15pt] (1,1)--(a);
    \end{tikzpicture}
  }
  &
  \subfigure[]{\label{fig:most-vide}
    \begin{tikzpicture}[xscale=.45,yscale=.45]
      \tikzstyle{every node}=[draw, circle, inner sep=1.2pt]
      \path (1,1) node (a) {};
      \path (a)+(18:1) node (b) {};
      \path (a)+(90:1) node (c) {};
      \path (a)+(162:1) node (d) {};
      \path (a)+(234:1) node (e) {};
      \path (a)+(306:1) node (f) {};
      \draw (a)--(b);
      \draw (a)--(c);
      \draw (a)--(d);
      \draw (a)--(e);
      \draw (a)--(f);
      \tikzstyle{every node}=[draw, fill, circle, inner sep=1.2pt]
      \path (4.5,1) node (a) {};
      \tikzstyle{every node}=[draw, dashed, dash pattern=on 1.2pt off 1.9pt, circle, inner sep=1.4pt]
      \path (a)+(18:1) node (b) {};
      \path (a)+(90:1) node (c) {};
      \path (a)+(162:1) node (d) {};
      \path (a)+(234:1) node (e) {};
      \path (a)+(306:1) node (f) {};
      \tikzstyle{every path}=[draw, dashed, dash pattern=on 1.6pt off 2.1pt]
      \draw (a)--(b);
      \draw (a)--(c);
      \draw (a)--(d);
      \draw (a)--(e);
      \draw (a)--(f);
      \tikzstyle{every path}=[]
      \draw[->, shorten >= 15pt, shorten <= 15pt] (1,1)--(a);
    \end{tikzpicture}
  }
  &
  \subfigure[]{\label{fig:us}
    \begin{tikzpicture}[xscale=.45,yscale=.45]
      \tikzstyle{every node}=[draw, circle, inner sep=1.2pt]
      \path (1,1) node (a) {};
      \path (a)+(1,0) node (b) {};
      \path (a)+(306:1.15) coordinate (space){};
      \tikzstyle{every node}=[draw, fill, circle, inner sep=1.2pt]
      \path (a)+(3,0) node (c) {};
      \path (a)+(4,0) node (d) {};
      \draw (a)--(b);
      \draw (c)--(d);
      \draw[->, shorten >= 5pt, shorten <= 5pt] (b)--(c);
    \end{tikzpicture}
  }
  &
  \subfigure[]{\label{fig:less}
    \begin{tikzpicture}[xscale=.45,yscale=.45]
      \tikzstyle{every node}=[draw, circle, inner sep=1.2pt]
      \path (1,1) node (a) {};
      \path (a)+(1,0) node (b) {};
      \path (a)+(306:1.15) coordinate (space){};
      \tikzstyle{every node}=[draw, fill, circle, inner sep=1.2pt]
      \path (a)+(3,0) node (c) {};
      \tikzstyle{every node}=[draw, dashed, dash pattern=on 1.2pt off 1.9pt, circle, inner sep=1.4pt]
      \path (a)+(4,0) node (d) {};
      \draw (a)--(b);
      \tikzstyle{every path}=[draw, dashed, dash pattern=on 1.6pt off 2.1pt]
      \draw (c)--(d);
      \tikzstyle{every path}=[]
      \draw[->, shorten >= 5pt, shorten <= 5pt] (b)--(c);
    \end{tikzpicture}
  }
  \end{tabular}
  \caption{Different scopes of local computations; the scope of $preconditions$ is depicted in \textit{white (on left sides)}, while the scope of $actions$ is depicted in \textit{black (on right sides)}. The \textit{dashed elements} represent entities (vertices or edges) that are considered by the preconditions but remain unchanged by the actions. }
  \label{fig:local_models}
  \end{center}
\end{figure}

The approach taken by local computations shares a number of traits with that of {\em population protocols}, more recently introduced in~\cite{AngluinADFP2006,angluin2007}. Both approaches work at a similar level of abstraction and are concerned with characterizing what can or cannot be done in distributed computing. As far as the {\em scope} of computation is concerned, population protocols can be seen as a particular case of local computation focusing on pairwise interaction (see Figure~\ref{fig:us}). The main difference between these tools (if any, besides that of originating from distinct lines of research), has more to do with the role given to the underlying synchronization between nodes. While local computations typically sees this as an lower layer being itself abstracted (whenever possible), population protocols consider the execution of an algorithm given some explicit properties of an {\em interaction scheduler}. This particularity led population protocols to become an appropriate tool to study distributed computing in dynamic networks, by reducing the network dynamics into specific properties of the scheduler ({\it e.g.,} every pair of nodes interact infinitely often). Several variants of population protocols have subsequently been introduced ({\it e.g.,} assuming various types of fairness of the scheduler and graphs of interaction), however we believe the analogy between dynamics and scheduling has some limits ({\it e.g.}, in reality two nodes that interact once will not necessarily interact twice; and the precise order in which a group of nodes interacts matters all the more when interactions do not repeat infinitely often). We advocate looking at the dynamics at a finer scale, without always assuming infinite recurrence on the scheduler (such a scheduler can still be formulated as a specific class of dynamics), in the purpose of studying the precise relationship between an algorithm and the dynamics underlying its execution. To remain as general as possible, we are building on top of local computations. One may ask whether remaining as general is relevant, and whether the various models on Figure~\ref{fig:local_models} are in fact equivalent in power ({\it e.g.} could we simulate any of them by repetition of another?). The answer is negative due to different levels of atomicity ({\it e.g.} models~\ref{fig:most} {\it vs.}~\ref{fig:us}) and symmetry breaking ({\it e.g.} models~\ref{fig:us} {\it vs.}~\ref{fig:less}). The reader is referred to~\cite{CMZ06} for a detailed hierarchy of these models. Note that the equivalences between models would have to be re-considered anyway in a dynamic context, since the dynamics may prevent the possibility of applying several steps of a weaker model to simulate a stronger one. 

We now describe the graph relabeling formalism traditionally associated with local computations. Let the network topology be represented by a finite undirected loopless graph $G = (V_G,E_G)$, with $V_G$ representing the set of nodes and $E_G$ representing the set of communication links between them. Two vertices $u$ and $v$ are said \textit{neighbors} if and only if they share a common edge $(u,v)$ in $E_G$. Let $\lambda:V_G \cup E_G \rightarrow \mathcal{L}^*$ be a mapping that associates every vertex and edge from $G$ with one or several labels from an alphabet $\mathcal{L}$ (which denotes all the possible states these elements can take). The state of a given vertex $v$, \textit{resp.} edge $e$, at a given time $t$ is denoted by $\lambda_t(v)$, \textit{resp.} $\lambda_t(e)$. The whole \textit{labeled graph} is represented by the pair $(G,\lambda)$, noted $\textbf{G}$.

According to~\cite{ac:litovsky99}, a complete algorithm can be given by a triplet {$\{\mathcal{L},\mathcal{I},P\}$}, where $\mathcal{I}$ is the set of initial states, and $P$ is a set of \emph{relabeling rules} (transition patterns) representing the distributed interactions --~these rules are considered {\it uniform} ({\it i.e.,} same for all nodes). The Algorithm~1 below ($\mathcal{A}_1$ for short), gives the example of a one-rule algorithm that represents the general broadcasting scheme discussed in the introduction. We assume here that the label $I$ (\textit{resp. $N$}) stands for the state \texttt{informed} (\textit{resp.} \texttt{non-informed}). Propagating the information thus consists in repeating this single rule, starting from the emitter vertex, until all vertices are labeled $I$.\footnote{Detecting such a final state is not part of the given algorithm. The reader interested in termination detection as a distributed problem is referred to~\cite{afif-detec}.}

\begin{algorithm}[h]
  \textit{\underline{initial states:}} {$\{I,N\}$ ($I$ for the initial emitter, $N$ for all the other vertices)}\\
  \textit{\underline{alphabet:}} {$\{I,N\}$}\\  
\begin{tabular}{cc|c}
  ~~
  &
  \begin{minipage}[c]{.48\linewidth}
    \vspace{2pt}
    \textit{\underline{preconditions($r_1$):}} $\lambda(v_0)=I \wedge \lambda(v_1)=N$
    \medbreak
    \textit{\underline{actions($r_1$):}} $\lambda(v_1):=I$\vspace{4pt}
  \end{minipage}
  &
  \begin{minipage}[c]{.40\linewidth}
    \textit{~\underline{graphical notation :}}\\
    ~\begin{tikzpicture}[scale=1]
      \tikzstyle{every node}=[draw, fill, circle, inner sep=1.5pt]
      \path (1,1) node (a) {};
      \path (a)+(1,0) node (b) {};
      \path (a)+(3,0) node (c) {};
      \path (a)+(4,0) node (d) {};
      \path (a)+(-0.5,0) coordinate (lanchor);
      \tikzstyle{every node}=[font=\scriptsize]
      \path (a)+(0,.25) node (la) {$I$};
      \path (b)+(0,.25) node (lb) {$N$};
      \path (c)+(0,.25) node (lc) {$I$};
      \path (d)+(0,.25) node (ld) {$I$};
      \draw (a)--(b);
      \draw (c)--(d);
      \draw[->, shorten >= 15pt, shorten <= 15pt] (b)--(c);
    \end{tikzpicture}
  \end{minipage}
\end{tabular}
\caption{A propagation algorithm coded by a single relabeling rule ($r_1$).}
\end{algorithm}

Let us repeat that an algorithm does not specify how the nodes synchronize, \textit{i.e.,} how they select each other to perform a common computation step. From the abstraction level of local computations, this underlying synchronization is seen as an implementation choice (dedicated procedures were designed to fit the various models, {\it e.g.} local elections~\cite{local-elections} and local rendezvous~\cite{rdv} for starwise and pairwise interactions, respectively). A direct consequence is that the execution of an algorithm at this level may not be deterministic. Another consequence is that the characterization of \emph{sufficient} conditions on the dynamics will additionally require assumptions on the synchronization --~we suggest later a generic progression hypothesis that serves this purpose. Note that the three algorithms provided in this paper rely on pairwise interactions, but the concepts and methodology involved apply to local computations in general. 


\subsection{Expressing dynamic network properties using Evolving Graphs}
\label{sec:evo}
In a different context, \textit{evolving graphs}~\cite{Fer04} were proposed as a
combinatorial model for dynamic networks. The initial
purpose of this model was to provide a suitable representation of
\textit{fixed schedule dynamic networks} (FSDNs), in order to compute optimal
communication routes such as shortest, fastest and foremost
journeys~\cite{BFJ03}. In such a context, the evolution of the network was
known beforehand. In the present work, we use evolving graphs in a very different purpose, which is to express properties on the network dynamics. It is important to keep in mind that the analyzed algorithms are never supposed to know the evolution of the network beforehand. 

An evolving graph is a structure in which the evolution of the network topology is recorded as a {\em sequence} of static graphs $\SG=G_1, G_2, ...$, where every $G_i=(V_i, E_i)$ corresponds to the network topology during an interval of time $[t_i, t_{i+1})$. Several {\em models} of dynamic networks can be captured by this {\em formalism}, depending on the meaning which is given to the sequence of dates $\ST=t_1, t_2,...$. For example, these dates could correspond to every time step in a {\em discrete-time} system (and therefore be taken from a time domain $\T\subseteq {\mathbb N}$), or to variable-size time intervals in {\em continuous-time} systems ($\T\subseteq {\mathbb R}$), where each $t_{i}$ is the date when a topological event occurs in the system ({\it e.g.}, appearance or disappearance of an edge in the graph), see for example Figure~\ref{fig:4-steps}.

We consider continuous-time evolving graphs in general. (Our results actually hold for any of the above meanings.) Formally, we consider an evolving graph as the structure $\EG$, where $G$ is the union of all $G_i$ in $\SG$, called the {\em underlying graph} of $\G$.
Henceforth, we will simply use the notations $V$ and $E$ to denote $V(G)$ and $E(G)$, the sets of vertices and edges of the underlying graph $G$. Since we focus here on computation models that are \emph{undirected}, we logically consider evolving graphs as being themselves undirected. The original version of evolving graphs considered \emph{undirected} edges, as well as possible restrictions on bandwidth and latency. Finally, we will use the notation $\G_{[t_a,t_b)}$ to denote the {\em temporal subgraph} $\G'=(G',\SG', \ST')$ built from $\EG$ such that $G'=G$, $\SG'=\{G_i \in \SG: t_i \in [t_a,t_b)\}$, and $\ST'=\{t_i \in \ST \cap [t_a, t_b)\}$.

\def\wgraph (#1,#2){%
  \tikzstyle{every node}=[draw,fill,circle,inner sep=#1]
  \path (0,0) node (c){};
  \path (c)+(-1,.4) node (a){};
  \path (c)+(1,.4) node (e){};
  \path (c)+(-.6,-.6) node (b){};
  \path (c)+(.6,-.6) node (d){};
  \tikzstyle{every node}=[font=\footnotesize]
  \path (a)+(-#2,0) node (la){$a$};
  \path (b)+(-#2,0) node (lb){$b$};
  \path[above] (c) node (lc){$c$};
  \path (d)+(#2,0) node (ld){$d$};
  \path (e)+(#2,0) node (le){$e$};
}
\begin{figure}[h]
  \centering
  \subfigure[Sequence of graphs and dates]{
    \label{fig:sequence}
    \begin{minipage}[c]{\linewidth}
      \centering
      \begin{tabular}{c|c|c|c}
        \hline
        period $t_0 \rightarrow t_1$&
        period $t_1 \rightarrow t_2$&
        period $t_2 \rightarrow t_3$&
        period $t_3 \rightarrow t_4$\\
        \begin{tikzpicture}[scale=.9]
          \wgraph(.8pt,6pt)
          \draw (a)--(c);
          \draw (b)--(c);
          \draw (b)--(d);
          \draw (c)--(d);
        \end{tikzpicture}
        &
        \begin{tikzpicture}[scale=.9]
          \wgraph(.8pt,6pt)
          \draw (a)--(b);
          \draw (b)--(c);
          \draw (c)--(d);
        \end{tikzpicture}
        &
        \begin{tikzpicture}[scale=.9]
          \wgraph(.8pt,6pt)
          \draw (a)--(b);
          \draw (c)--(d);
          \draw (c)--(e);
          \draw (d)--(e);
        \end{tikzpicture}
        &
        \begin{tikzpicture}[scale=.9]
          \wgraph(.8pt,6pt)
          \draw (c)--(e);
          \draw (d)--(e);
        \end{tikzpicture}
        \\
        $G_0$ & $G_1$ & $G_2$ & $G_3$\\\hline
      \end{tabular}
      \vspace{5pt}
    \end{minipage}
  }
  \subfigure[A compact representation]{
    \label{fig:compact}
    \begin{tikzpicture}[scale=2]
      \wgraph(1pt,3pt)
      \tikzstyle{every node}=[sloped,below,font=\scriptsize,inner sep=2pt]
      \draw (a)--node{$[t_1,t_3)$}(b);
      \draw (b)--node{$[t_0,t_1)$}(d);
      \draw (d)--node{$[t_2,t_4)$}(e);
      \tikzstyle{every node}=[sloped,above,font=\scriptsize,inner sep=1pt]
      \draw (a)--node[inner sep=1.5pt]{$[t_0,t_1)$}(c);
      \draw (b)--node{$[t_0,t_2)$}(c);
      \draw (c)--node{$[t_0,t_3)$}(d);
      \draw (c)--node{$[t_2,t_4)$}(e);
      \tikzstyle{every node}=[font=\large]
      \path (-1.4,-.1) node (G){$\G=$};
      \path (0,.6) coordinate (toplimit);
      \path (1.5,0) coordinate (rightlimit);
    \end{tikzpicture}
  }
  \caption{\label{fig:4-steps}Example of evolving graph.}
\end{figure}
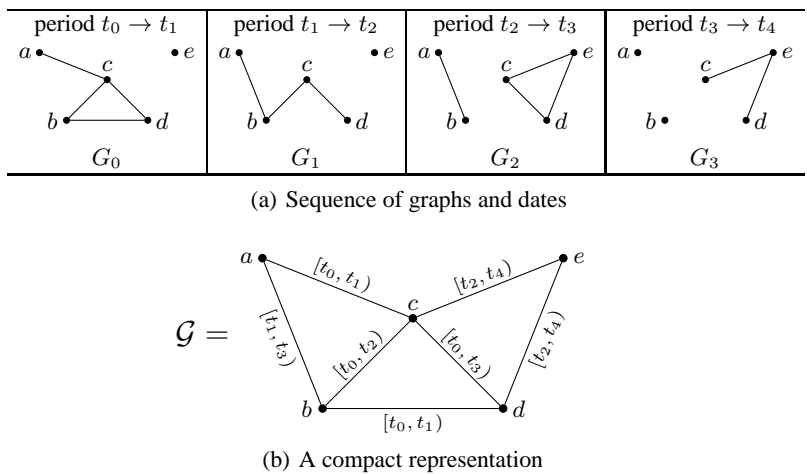

\subsection{Basic concepts and notations (given an evolving graph $\EG$).}

As a writing facility, we consider the use of a {\em presence function} $\rho: E \times \T \rightarrow \{0,1\}$ that indicates whether a given edge is present at a given date, that is, for $e \in E$ and $t\in [t_i, t_{i+1})$ (with $t_i, t_{i+1} \in \ST$), $\rho(e,t)=1 \iff e\in E_i$.

A central concept in dynamic networks is that of {\em journey}, which is the {\em temporal} extension of the concept of path. A journey can be thought of as a path {\em over time} from one vertex to another. Formally, a sequence of couples $\J=\{(e_1,\sigma_1),$ $(e_2,\sigma_2) \dots,$ $(e_k,\sigma_k)\}$ such that $\{e_1, e_2,...,e_k\}$ is a  walk in $G$ and $\{\sigma_1, \sigma_2,..., \sigma_k\}$ is a non-decreasing sequence of dates from $\T$, is a {\em journey} in $\G$ if and only if $\rho(e_i,\sigma_i)=1$ for all $i\le k$. We will say that a given journey is {\em strict} if every couple $(e_i,\sigma_i)$ is taken from a distinct graph of the sequence $\SG$.

Let us denote by $\J^*$ the set of all possible journeys in an evolving graph $\G$, and by  $\J^{*}_{(u,v)} \subseteq \J^*$ those journeys starting at node $u$ and ending at node $v$. 
If a journey exists from a node $u$ to a node $v$, that is, if $\J^{*}_{(u,v)} \ne \emptyset$, then we say that $u$ can {\em reach} $v$ in a graph $\G$, and allow the simplified notations $u\leadsto v$ (in $\G$), or $u \leadstoo v$ if this can be done through a strict journey. Clearly, the existence of journey is not symmetrical: $u \leadsto v \nLeftrightarrow v \leadsto u$; this holds regardless of whether the edges are directed or not, because the time dimension creates its own level of direction -- this point is clear by the example of Figure~\ref{fig:abc}. Given a node $u$, the set $\{v\in V : u \leadsto v\}$ is called the {\em horizon} of $u$. We assume that every node belongs to its own horizon by means of an empty journey. Here are examples of journeys in the evolving graph of Figure~\ref{fig:4-steps}:
\begin{itemize}
\item $J_{(a,e)}$=$\{(ab,\sigma_1 \in [t_1,t_2)),(bc,\sigma_2 \in [\sigma_1,t_2)),(ce,\sigma_3 \in [t_2,t_3))\}$ is a journey from $a$ to $e$\,;
\item $J_{(a,e)}$=$\{(ac,\sigma_1 \in [t_0,t_1)),(cd,\sigma_2 \in [\sigma_1,t_1),(de,\sigma_3 \in [t_3,t_4))\}$ is another journey from $a$ to $e$\,;
\item $J_{(a,e)}$=$\{(ac,\sigma_1 \in [t_0,t_1)),(cd,\sigma_2 \in [t_1,t_2),(de,\sigma_3 \in [t_3,t_4))\}$ is yet another ({\em strict}) journey from $a$ to $e$.
\end{itemize}

We will say that the network is {\em connected over time} iff $\forall u,v \in V, u \leadsto v \wedge v \leadsto u$. The concept of connectivity over time is not new and goes back at least to~\cite{awerbuch84}, in which it was called {\em eventual connectivity} (although recent literature on DTNs referred to this terms for another concept that we renamed {\em eventual instant-connectivity} to avoid confusion in Section~\ref{sec:classification}).

\section{The proposed analysis framework}
\label{sec:framework}
As a recall of the previous section, the algorithmic state of the network is given by a labeling on the corresponding graph $G$, then noted $\textbf{G}$. We denote by $G_i$ the graph covering the period $[t_i,t_{i+1})$ in the evolving graph $\G=(G,\SG,\ST)$, with $G_i \in \SG$ and $t_i,t_{i+1} \in \ST$. Notice that the symbol $G$ was used here with two different meanings: the first as the generic letter to represent the network, the second to denote the \emph{underlying graph} of $\G$. Both notations are kept as is in the following, while preventing ambiguous uses in the text.

\subsection{Putting the pieces together: relabelings over evolving graphs}

For an evolving graph $\G=(G,\SG,\ST)$ and a given date $t_i \in \ST$, we denote by \textbf{G}$_{i}$ the labeled graph $(G_i,\lambda_{t_i+\epsilon})$ representing the state of the network {\em just after} the topological event of date $t_i$, and by \textbf{G}$_{i[}$ the labeled graph $(G_{i-1},\lambda_{t_{i}-\epsilon})$ representing the network state {\em just before} that event. We note
\begin{equation*}
  Event_{t_i}(\textbf{G}_{i[}) = \textbf{G}_i~.
\end{equation*}
A number of distributed operations may occur between two consecutive events. Hence, for a given algorithm~$\mathcal{A}$ and two consecutive dates $t_i,t_{i+1} \in \ST$, we denote by $\mathcal{R}_{\mathcal{A}_{[t_i,t_{i+1})}}$ one of the possible relabeling sequence induced by $\mathcal{A}$ on the graph $G_i$ during the period $[t_i,t_{i+1})$. We note
\begin{equation*}
  \mathcal{R}_{\mathcal{A}_{[t_i,t_{i+1})}}(\textbf{G}_{i})~=~\textbf{G}_{i+1[}~.
\end{equation*}
For simplicity, we will sometimes use the notation $r_i(u,v) \in \mathcal{R}_{\mathcal{A}_{[t,t')}}$ to indicate that the rule $r_i$ is applied on the edge $(u,v)$ during $[t,t')$. A complete execution sequence from $t_0$ to $t_{k}$ is then given by means of an alternated sequence of relabeling steps and topological events, which we note
\begin{equation*}
  X\mbox{=}\mathcal{R}_{\mathcal{A}_{[t_{k-1},t_{k})}} \circ Event_{_{t_{k-1}}} \circ\mbox{..}\circ Event_{_{t_i}} \circ \mathcal{R}_{\mathcal{A}_{[t_{i-1},t_{i})}} \circ\mbox{..}\circ Event_{_{t_1}} \circ \mathcal{R}_{\mathcal{A}_{[t_0,t_1)}} (\textbf{G}_{0})
\end{equation*}
This combination is illustrated on Figure~\ref{fig:combination}.
As mentioned at the end of Section~\ref{sec:grs}, the execution of a local computation algorithm is not necessarily deterministic, and may depend on the way nodes select one another at a lower level before applying a relabeling rule. Hence, we denote by {\large $\mathcal{X}_{\mathcal{A}/\G}$} the set of all possible execution sequences of an algorithm $\mathcal{A}$ over an evolving graph $\G$. 

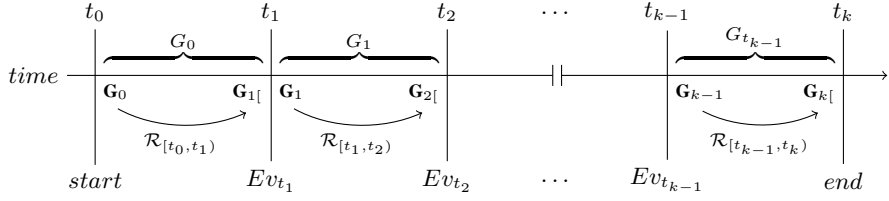
\begin{figure}[h]
  \begin{center}
    \def\grbar (#1,#2,#3){
      \tikzstyle{every node}=[font=\footnotesize]
      \path (#1,0) node (bas){#2};
      \path (#1,3.2) node (haut){#3};
      \draw (bas)--(haut);
    }
    \def\grinter (#1,#2,#3,#4,#5){
      \tikzstyle{every node}=[font=\scriptsize]
      \path (#1,2) coordinate (cross1);
      \path (#1,2)+(2,0) coordinate (cross2);
      \path[below right] (cross1) node (G1){#2};
      \path[below left] (cross2) node (G2){#3};
      \path (G1.south) node[inner sep=0] (G1S){};
      \path (G2.south) node[inner sep=0] (G2S){};
      \draw[->] (G1S) edge[bend right=40] node[midway,below](R1){#4} (G2S);
      \path[above] (#1,2)+(1,0) node (Br1){$\overbrace{\hspace{60pt}}$};
      \path[above] (Br1) node (GN1){#5};
    }
    
    \begin{tikzpicture}[xscale=1.17, yscale=.7]
      \path (0.3,2) node (ntime){$time$};
      \draw (ntime)--(6.2,2); \draw[->](6.3,2)--(10,2);
      \draw (6.2,1.8)--(6.2,2.2);
      \draw (6.3,1.8)--(6.3,2.2);
      \grbar (1,$start$,$t_0$)
      \grinter (1,\textbf{G}$_0$,\textbf{G}$_{1[}$,$\mathcal{R}_{[t_0,t_1)}$,$G_0$)
      \grbar (3,$Ev_{t_1}$,$t_1$)
      \grinter (3,\textbf{G}$_1$,\textbf{G}$_{2[}$,$\mathcal{R}_{[t_1,t_2)}$,$G_1$)
      \grbar (5,$Ev_{t_2}$,$t_2$)
      \grbar (7.5,$Ev_{t_{k-1}}$,$t_{k-1}$)
      \grinter (7.5,\textbf{G}$_{k-1}$,\textbf{G}$_{k[}$,$\mathcal{R}_{[t_{k-1},t_{k})}$,$G_{t_{k-1}}$)
      \grbar (9.5,$end$,$t_{k}$)
      \path (6.25,0) node (bas){$\dots$};
      \path (6.25,3.2) node (haut){$\dots$};      
    \end{tikzpicture}
  \end{center}
  \caption{\label{fig:combination}Combination of Graph Relabelings and Evolving Graphs.}
\end{figure}

\subsection{Methodology}
\label{sec:char-tools}
Below are some proposed methods and concepts to characterize the requirement of an algorithm in terms of topology dynamics. More precisely, we use the above combination to define the concept of topology-related \emph{necessary} or \emph{sufficient} conditions, and discuss how a given property can be proved to be so.

\subsubsection{Objectives of an algorithm}
Given an algorithm $\mathcal{A}$ and a labeled graph \textbf{G}, the state one wishes to reach can be given by a logic formula $\mathcal{P}$ on the labels of vertices (and edges, if appropriate). In the case of the propagation scheme (Algorithm~1 Section~\ref{sec:grs}), such a terminal state could be that all nodes are informed,
\begin{equation*}
  \mathcal{P}_1(\textbf{G})=\forall v \in V, \lambda(v)=I.
\end{equation*}

The objective ${\cal O}_{\cal A}$ is then defined as the fact of verifying the desired property by the end of the execution, that is, on the final labeled graph ${\bf G}_{k}$. In this example, we consider $\mathcal{O}_{\mathcal{A}_1}=\mathcal{P}_1(\textbf{G}_{k})$. The opportunity must be taken here to talk about two fundamentally different types of objectives in dynamic networks. In the example above, as well as in the other examples in this paper, we consider algorithms whose objective is to {\em reach} a given property by the end of the execution. Another type of objective in dynamic network is to consider the \emph{maintenance} of a desired property despite the network evolution (e.g. covering every connected component in the network by a single spanning tree). In this case, the objective must not be formulated in terms of terminal state, but rather in terms of satisfactory state, for example in-between every two consecutive topological events, {\em i.e.,} $\mathcal{O}_{\mathcal{A}}=\forall G_{i}\in S_{G}, \mathcal{P}(\textbf{G}_{i+1[})$. This actually corresponds to a {\em self-stabilization} scenario where recurrent faults are the topological events, and the network must stabilize in-between any two consecutive faults. We restrict ourselves to the first type of objective in the following.

\subsubsection{Necessary conditions}

Given an algorithm $\mathcal{A}$, its objective $\mathcal{O}_{\mathcal{A}}$ and an evolving graph property $\mathcal{C}_{\scriptsize \mathcal{N}}$, the property $\mathcal{C}_{\scriptsize \mathcal{N}}$ is a \emph{(topology-related) necessary} condition for $\mathcal{O}_{\mathcal{A}}$ if and only if
\begin{equation*}
  \forall \G, \neg \mathcal{C}_{\mathcal N}(\G) \implies \neg \mathcal{O}_{\mathcal{A}}
\end{equation*}
\noindent Proving this result comes to prove that {$\forall \G, \neg$~$ \mathcal{C}_{\mathcal N}(\G) \implies \nexists X \in \mathcal{X}_{\mathcal{A}/\G} \mid \mathcal{P}($\textbf{G}$_{k})$}. (The desired state is not reachable by the end of the execution (time $k$), unless the condition is verified.)

\subsubsection{Sufficient conditions}

Symmetrically, an evolving graph property $\mathcal{C}_{\mathcal S}$ is a \emph{(topology-related) sufficient} condition for $\mathcal{A}$ if and only if
\begin{equation*}
  \forall \G, \mathcal{C}_{\mathcal S}(\G) \implies \mathcal{O}_{\mathcal{A}}
\end{equation*}
\noindent Proving this result comes to prove that $\forall \G, \mathcal{C}_{\mathcal S}(\G) \implies \forall X \in \mathcal{X}_{\mathcal{A}/\G}, \mathcal{P}($\textbf{G}$_{k})$.\medskip

Because the abstraction level of these computations is not concerned with the underlying synchronization, no topological property can guarantee, alone, that the nodes will effectively communicate and collaborate to reach the desired objective. Therefore, the characterization of {\em sufficient} conditions requires additional assumptions on the synchronization. We propose below a generic progression hypothesis applicable to the pairwise interaction model (Figure~\ref{fig:us}). This assumption may or may not be considered realistic depending on the expected rate of topological changes.
\begin{hypo}($PH_1$).\label{hyp:prog} In every time interval $[t_i,t_{i+1})$, with $t_i$ in $\ST$, each vertex is able to apply at least one relabeling rule with each of its neighbors, provided the rule preconditions are \emph{already} satisfied at time $t_i$ (and still satisfied at the time the rule is applied).
\end{hypo}

In the case when starwise interaction (see Figure~\ref{fig:most-vide}) is considered, this hypothesis could be partially relaxed to assuming only that every node applies at least one rule in each interval.

\section{Examples of basic analyses}
\label{sec:analysis}
This section illustrates the proposed framework through the analysis of three basic algorithms, namely the propagation algorithm previously given, and two counting algorithms (one centralized, one decentralized). The results obtained here are used in the next section to highlight some implications of this work.  

\subsection{Analysis of the propagation algorithm}
We want to prove that the existence of a journey (\textit{resp.} strict journey) between the emitter and every other node is a necessary (\textit{resp.} sufficient) condition to achieve $\mathcal{O}_{\mathcal{A}_1}$. Our purpose is not as much to emphasize the results themselves --~they are rather intuitive~-- as to illustrate how the characterizations can be written in a rigorous way.
  
\begin{cond}
$\forall v \in V, emitter \leadsto v$\\(There exists a journey between the emitter and every other vertex).
\end{cond}

\addtocounter{prop}{1}
\begin{lemme} 
  \label{lem1}
  $\forall v \in V : \lambda_{t_0}(v)=N, \lambda_{\sigma>t_0}(v)=I \implies \exists u \in V, \exists \sigma' \in [t_0, \sigma) : \lambda_{\sigma'}(u)=I \wedge u \leadsto v~in~\G_{[\sigma',\sigma)}$\\
(If a non-emitter vertex has the information at some point, it implies the existence of an incoming journey from a vertex that had the information before)
\end{lemme}
\addtocounter{prop}{-1}
\begin{proof}
  $\forall v \in V : \lambda_{t_0}(v) = N, (\lambda_{\sigma>t_0}(v) = I\implies\exists v' \in V : r_{1(v',v)} \in \mathcal{R}_{\mathcal{A}_1[t_0,\sigma)})$\\
  (If a non-emitter vertex has the information at some point, then it has necessarily applied rule $r_1$ with another vertex)\\
  $\implies \exists v' \in V, \sigma' \in [t_0, \sigma) : \lambda_{\sigma'}(v')=I \wedge \rho((v',v),\sigma')=1$\\
  (An edge existed at a previous date between this vertex and a vertex labeled $I$)\\
  By transitivity, $\implies \exists v'' \in V, \exists \sigma'' \in [t_0, \sigma) : \lambda_{\sigma''}(v'')=I \wedge v'' \leadsto v~in~\G_{[\sigma'',\sigma)}$\\  
  (A journey existed between a vertex labeled $I$ and this vertex)
\carre
\end{proof}

\begin{prop}
  Condition 1 ($\mathcal{C}_1$) is a necessary condition on $\G$ to allow Algorithm~1 ($\mathcal{A}_1$) to reach its objective $\mathcal{O}_{\mathcal{A}_1}$.
\end{prop}

\begin{proof}
(using Lemma~\ref{lem1}).
  Following from Lemma~\ref{lem1} and the initial states ($I$ for the emitter, $N$ for all other vertices), we have $\mathcal{O}_{\mathcal{A}_1} \implies \mathcal{C}_1$, and thus $\neg \mathcal{C}_1 \implies \neg \mathcal{O}_{\mathcal{A}_1}$
\carre
\end{proof}

\begin{cond}
$\forall v \in V, emitter \leadstoo v$
\end{cond}

\begin{prop}
    Under Progression Hypothesis~\ref{hyp:prog} ($PH_1$, defined in the previous section), Condition~2 ($\mathcal{C}_2$) is sufficient on $\G$ to guarantee that $\mathcal{A}_1$ will reach $\mathcal{O}_{\mathcal{A}_1}$.
\end{prop}

\begin{proof}
  (1): By $PH_1$, $\forall t_i \in \mathcal{S}_{\Bbb T}{\setminus(t_{k})}, \forall (u,u') \in E_i,\lambda_{t_i}(u) = I \implies \lambda_{t_{i+1}}(u') = I$ \\
  By iteration on (1): $\forall u,v \in V, u \leadstoo v \implies (\lambda_{t_0}(u)$=$I \implies \lambda_{t_{k}}(v)$=$I)$\\
  Now, because $\lambda_{t_0}(emitter)=I$, we have $\mathcal{C}_2(\G) \implies \forall X \in \mathcal{X}_{\mathcal{A}/\G}, \mathcal{P}_1($\textbf{G}$_{k})$
\carre
\end{proof}

\subsection{Analysis of a centralized counting algorithm}

Like the propagation algorithm, the distributed algorithm presented below assumes a distinguished vertex at initial time. This vertex, called the \textit{counter}, is in charge of counting all the vertices it meets during the execution (its successive neighbors in the changing topology). Hence, the counter vertex has two labels $(C,i)$, meaning that it is the counter ($C$), and that it has already counted $i$ vertices (initially $1$, \textit{i.e.,} itself). The other vertices are labeled either $F$ or $N$, depending on whether they have already been counted or not. The counting rule is given by $r_1$ in Algorithm~2, below.

\begin{algorithm}[H]~\\
  \textit{\underline{initial states:}} {$\{(C,1),N\}$ ($(C,1)$ for the counter, $N$ for all other vertices)}\\
  \textit{\underline{alphabet:}} {$\{C,N,F,\mathbb{N}^*\}$}\\
  \textit{\underline{rule $r_1$:}}\\
    \begin{tikzpicture}[scale=1]
      \tikzstyle{every node}=[draw, fill, circle, inner sep=1.5pt]
      \path (-.5,1) coordinate (lanchor);
      \path (1,1) node (a) {};
      \path (a)+(1,0) node (b) {};
      \path (a)+(3,0) node (c) {};
      \path (a)+(4,0) node (d) {};
      \tikzstyle{every node}=[font=\scriptsize]
      \path (a)+(0,.25) node (la) {$C,i$};
      \path (b)+(0,.25) node (lb) {$N$};
      \path (c)+(0,.25) node (lc) {$C,i+1$};
      \path (d)+(0,.25) node (ld) {$F$};
      \draw (a)--(b);
      \draw (c)--(d);
      \draw[->, shorten >= 15pt, shorten <= 15pt] (b)--(c);
    \end{tikzpicture}\medbreak
\caption{Counting algorithm with a pre-selected counter.}
\end{algorithm}

\paragraph{Objective of the algorithm.}
Under the assumption of a fixed number of vertices, the algorithm reaches a terminal state when all vertices are counted, which corresponds to the fact that no more vertices are labeled $N$:
\begin{equation*}
  \mathcal{P}_2=\forall v \in V, \lambda(v)\ne N
\end{equation*}
The objective of Algorithm~2 is to satisfy this property at the end of the execution ($\mathcal{O}_{\mathcal{A}_2}=\mathcal{P}_2(\textbf{G}_{k})$). We prove here that the existence of an edge at some point of the execution between the \emph{counter} node and every other node is a necessary and sufficient condition.

\begin{cond} 
  $\forall v \in V{\setminus{\{counter\}}}, \exists t_i \in \ST : (counter,v) \in E_i$, or equivalently with the notion of underlying graph, $\forall v \in V{\setminus\{counter\}}, (counter,v) \in E$
\end{cond}

\begin{prop}
  For a given evolving graph $\G$ representing the topological evolutions that take place during the execution of $\mathcal{A}_2$, Condition~3 ($\mathcal{C}_3$) is a \textup{necessary} condition on $\G$ to allow $\mathcal{A}_2$ to reach its objective $\mathcal{O}_{\mathcal{A}_2}$.
\end{prop}

\begin{proof}
  $\neg \mathcal{C}_3 (\G) \implies  \exists v \in V{\setminus \{counter\}} : (counter,v) \notin E$\\
  $\implies \exists v \in V{\setminus \{counter\}} : \forall t_i \in \mathcal{S}_\mathbb{T}{\setminus\{t_{k}\}}, r_1(counter,v) \notin \mathcal{R}_{\mathcal{A}_2{[t_i,t_{i+1})}}$\\
  $\implies \exists v \in V{\setminus \{counter\}} : \forall X \in \mathcal{X}_{\mathcal{A}_2/\G}, \lambda_{t_{k}}(v) = N$\\
  $\implies \nexists X \in \mathcal{X}_{\mathcal{A}_2/\G} : \mathcal{P}_2($\textbf{G}$_{k})\implies \neg \mathcal{O}_{\mathcal{A}_2}$
\carre
\end{proof}

\begin{prop}
  Under Progression Hypothesis~\ref{hyp:prog} (noted $PH_1$ below), $\mathcal{C}_3$ is also a sufficient condition on $\G$ to guarantee that $\mathcal{A}_2$ will reach its objective $\mathcal{O}_{\mathcal{A}_2}$.
\end{prop}

\begin{proof}
    $\mathcal{C}_3(\G)\implies \forall v \in V{\setminus\{counter\}}, \exists t_i \in \mathcal{S}_{\Bbb T} : (counter,v) \in E_i$\\
    by $PH_1$, $\implies \forall v \in V{\setminus\{counter\}}, \exists t_i \in \mathcal{S}_\mathbb{T} : r_1(counter,v) \in \mathcal{R}_{\mathcal{A}_2[t_i,t_{i+1})}$\\
    $\implies \forall v \in V{\setminus\{counter\}}, \lambda_{t_{k}}(v)\ne N$\\
    $\implies \forall X \in \mathcal{X}_{\mathcal{A}_2/\G}, \mathcal{P}_2(\textbf{G}_{k})\implies \mathcal{O}_{\mathcal{A}_2}$
\carre
\end{proof}

\subsection{Analysis of a decentralized counting algorithm}

Contrary to the previous algorithm, Algorithm~3 below does not require a distinguished initial state for any vertex. Indeed, all vertices are initialized with the same labels $(C,1)$, meaning that they are all initially counters that have already included themselves into the count. Then, depending on the topological evolutions, the counters opportunistically merge by pairs (rule $r_1$) in Algorithm~$\mathcal{A}_3$. In the optimistic scenario, at the end of the execution, only one node remains labeled $C$ and its second label gives the total number of vertices in the graph. A similar counting principle was used in~\cite{AngluinADFP2006} to illustrate population protocols --~a possible application of this protocol was anecdotally mentioned, consisting in monitoring a flock of birds for fever, with the role of {\em counters} being played by sensors.

\begin{algorithm}[H]~\\
  \textit{\underline{initial states:}} {$\{(C,1)\}$ (for all vertices)} \\\textit{\underline{alphabet:}} {$\{C,F,\mathbb{N}^*\}$}\\
  \textit{\underline{rule $r_1$:}}\\  
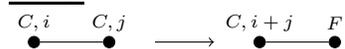
    \begin{tikzpicture}[scale=1]
      \tikzstyle{every node}=[draw, fill, circle, inner sep=1.5pt]
      \path (1,1) node (a) {};
      \path (a)+(1,0) node (b) {};
      \path (a)+(3,0) node (c) {};
      \path (a)+(4,0) node (d) {};
      \tikzstyle{every node}=[font=\scriptsize]
      \path (a)+(0,.25) node (la) {$C,i$};
      \path (b)+(0,.25) node (lb) {$C,j$};
      \path (c)+(0,.25) node (lc) {$C,i+j$};
      \path (d)+(0,.25) node (ld) {$F$};
      \draw (a)--(b);
      \draw (c)--(d);
      \draw[->, shorten >= 15pt, shorten <= 15pt] (b)--(c);
    \end{tikzpicture}\medbreak
\caption{Decentralized counting algorithm.}
\end{algorithm}

\subsubsection{Objective of the algorithm}
Under the assumption of a fixed number of vertices, this algorithm reaches the desired state when exactly one vertex remains labeled $C$:
\begin{equation*}
  \mathcal{P}_3 = \exists u \in V : \forall v \in V{\setminus \{u\}}, \lambda(u)=C \wedge \lambda(v) \ne C.
\end{equation*}

As with the two previous algorithms, the objective here is to reach this property by the end of the execution: $\mathcal{O}_{\mathcal{A}_3}=\mathcal{P}_3(\textbf{G}_{k})$. The characterization below proves that the existence of a vertex belonging to the {\em horizon} of every other vertex is a necessary condition for this algorithm.

\begin{cond}
  $\exists v \in V : \forall u \in V, u \leadsto v$
\end{cond}

\addtocounter{prop}{1}

\begin{lemme}\label{lem2}
  $\forall u \in V, \exists u' \in V : u\leadsto u' \wedge \lambda_{t_k}(u')=C$\\
(Counters cannot disappear from their own horizon.)
\end{lemme}

\noindent This lemma is proven in natural language because the equivalent steps would reveal substantially longer and inelegant (at least, without introducing further notations on sequences of relabelings). One should however see without effort how the proof could be technically translated.

\begin{proof}
  (by contradiction). The only operation that can suppress $C$ labels is the application of $r_1$. Since all vertices are initially labeled $C$, assuming that Lemma~\ref{lem2} is false (i.e., that there is no $C$-labeled vertex in the horizon of a vertex) comes to assume that a relabeling sequence took place transitively from vertex $u$ to a vertex $u'$ that is outside the horizon of $u$, which is by definition impossible.\carre
\end{proof}


\addtocounter{prop}{-1}
\begin{prop}
  Condition~4 ($\mathcal{C}_4$) is \textup{necessary} for $\mathcal{A}_3$ to reach its objective~$\mathcal{O}_{\mathcal{A}_3}$.
\end{prop}

\begin{proof}
  $\neg \mathcal{C}_4(\G)\implies \nexists v \in V : \forall u \in V, u \leadsto v$\\
  $\implies \forall v \in V : \lambda_{t_{k}}(v)=C, \exists u \in V : u \not \leadsto v$\\
  \textit{(Given any final counter, there is a vertex that could not reach it by a journey)}.\\
  By Lemma~\ref{lem2},  $\implies \forall v \in V : \lambda_{t_{k}}(v)=C, \exists v' \in V{\setminus\{v\}} : \lambda_{t_{k}}(v')=C$\\\textit{(There are at least two final counters)}.\\
  $\implies \neg \mathcal{P}_3(\textbf{G}_{k})\implies \neg \mathcal{O}_{\mathcal{A}_3}$
  \carre
\end{proof}

\noindent The characterization of a sufficient condition for $\mathcal{A}_3$ is left open. This question is addressed from a probabilistic perspective in~\cite{AngluinADFP2006}, but we believe a deterministic condition should also exist, although very specific.

\section{Classification of dynamic networks and algorithms}
\label{sec:classification}
In this section, we show how the previously characterized conditions can be used to define evolving graph classes, some of which are included in others. The relations of inclusion lead to a \emph{de facto} classification of dynamic networks based on the properties they verify. As a result, the classification can in turn be used to compare several algorithms or problems on the basis of their topological requirements. Besides the classification based on the above conditions, we discuss a possible extension of $10$ more classes considered in various recent works. 

\subsection{From conditions to classes of evolving graphs}
\label{ssec:clas}

From $\mathcal{C}_1= \forall v \in V, emitter \leadsto v$, we derive two classes of evolving graphs. $\mathcal{F}_1$ is the class in which at least one vertex can reach all the others by a journey. If an evolving graph does not belong to this class, then there is no chance for $\mathcal{A}_1$ to succeed whatever the initial emitter. $\mathcal{F}_2$ is the class where every vertex can reach all the others by a journey. If an evolving graph does not belong to this class, then at least one vertex, if chosen as an initial emitter, will fail to inform all the others using $\mathcal{A}_1$.\medskip

From $\mathcal{C}_2= \forall v \in V, emitter \leadstoo v$, we derive two classes of evolving graphs. $\mathcal{F}_3$ is the class in which at least one vertex can reach all the others by a {\em strict} journey. If an evolving graph belongs to this class, then there is at least one vertex that could, for sure, inform all the others using $\mathcal{A}_1$ (under Progression Hypothesis~\ref{hyp:prog}). $\mathcal{F}_4$ is the class of evolving graphs in which every vertex can reach all the others by a {\em strict} journey. If an evolving graph belongs to this class, then the success of $\mathcal{A}_1$ is guaranteed for any vertex as initial emitter (again, under Progression Hypothesis~\ref{hyp:prog}). \medskip

From $\mathcal{C}_3=\forall v \in V{\setminus\{counter\}}, (counter,v) \in E$, we derive two classes of graphs. $\mathcal{F}_5$ is the class of evolving graphs in which at least one vertex shares, at some point of the execution, an edge with every other vertex. If an evolving graph does not belong to this class, then there is no chance of success for $\mathcal{A}_2$, whatever the vertex chosen for counter. Here, if we assume Progression Hypothesis~\ref{hyp:prog}, then $\mathcal{F}_5$ is also a class in which the success of the algorithm can be guaranteed for one specific vertex as counter. $\mathcal{F}_6$ is the class of evolving graphs in which every vertex shares an edge with every other vertex at some point of the execution. If an evolving graph does not belong to this class, then there exists at least one vertex that cannot count all the others using $\mathcal{A}_2$. Again, if we consider Progression Hypothesis~\ref{hyp:prog}, then $\mathcal{F}_6$ becomes a class in which the success is guaranteed whatever the counter. \medskip

Finally, from $\mathcal{C}_4=\exists v \in V : \forall u \in V, u \leadsto v$, we derive the class $\mathcal{F}_7$, which is the class of graphs such that at least one vertex can be reached from all the others by a journey (in other words, the intersection of all nodes {\em horizons} is non-empty). If a graph does not belong to this class, then there is absolutely no chance of success for $\mathcal{A}_3$.

\subsection{Relations between classes}
Since \textit{all} implies \textit{at least one}, we have:  $\mathcal{F}_2 \subseteq \mathcal{F}_1$, $\mathcal{F}_4 \subseteq \mathcal{F}_3$, and  $\mathcal{F}_6 \subseteq \mathcal{F}_5$. Since a strict journey is a journey, we have:  $\mathcal{F}_3 \subseteq \mathcal{F}_1$, and  $\mathcal{F}_4 \subseteq \mathcal{F}_2$. Since an edge is a (strict) journey, we have:  $\mathcal{F}_5 \subseteq \mathcal{F}_3$, $\mathcal{F}_6 \subseteq \mathcal{F}_4$, and  $\mathcal{F}_5 \subseteq \mathcal{F}_7$. Finally, the existence of a journey between all pairs of vertices ($\mathcal{F}_2$) implies that each vertex can be reached by all the others, which implies in turn that at least one vertex can be reach by all the others ( $\mathcal{F}_7$). We then have: $\mathcal{F}_2 \subseteq \mathcal{F}_7$. Although we have used here a non-strict inclusion ($\subseteq$), the inclusions described above are strict (one easily find for each inclusion a graph that belongs to the parent class but is outside the child class). Figure~\ref{fig:classif} summarizes all these relations.

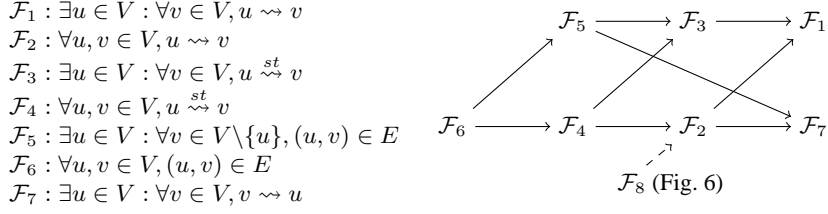
\begin{figure}[h]
  \begin{center}
    \begin{tabular}{cc}
      \begin{minipage}[c]{5.5cm}
        $\mathcal{F}_1:\exists u \in V : \forall v \in V, u \leadsto v$\\
        $\mathcal{F}_2:\forall u,v \in V, u \leadsto v$\\
        $\mathcal{F}_3:\exists u \in V : \forall v \in V, u \leadstoo v$\\
        $\mathcal{F}_4:\forall u,v \in V, u \leadstoo v$\\
        $\mathcal{F}_5:\exists u \in V : \forall v \in V{\setminus \{u\}}, (u,v) \in E$\\
        $\mathcal{F}_6:\forall u,v \in V, (u,v) \in E$\\
        $\mathcal{F}_7:\exists u \in V : \forall v \in V, v \leadsto u$
      \end{minipage}
      &
      \begin{minipage}[c]{5cm}
      \begin{tikzpicture}[xscale=1.6, yscale=1.4]
        \tikzstyle{every node}=[]
        \path (0,0) node (f6) {$\mathcal{F}_6$};
        \path (1,0) node (f4) {$\mathcal{F}_4$};
        \path (1,1) node (f5) {$\mathcal{F}_5$};
        \path (2,0) node (f2) {$\mathcal{F}_2$};
        \path (2,1) node (f3) {$\mathcal{F}_3$};
        \path (3,0) node (f7) {$\mathcal{F}_7$};
        \path (3,1) node (f1) {$\mathcal{F}_1$};
        \path (1.81,-.55) node (fig) {$\mathcal{F}_8$ \footnotesize (Fig.~\ref{fig:classif2})};
        \tikzstyle{every path}=[->]
        \draw (f6)--(f4);
        \draw (f6)--(f5);
        \draw (f4)--(f2);
        \draw (f4)--(f3);
        \draw (f5)--(f3);
        \draw (f5)--(f7);
        \draw (f2)--(f7);
        \draw (f2)--(f1);
        \draw (f3)--(f1);
        \draw[dashed] (fig.north)+(-.2,0)--(f2);
      \end{tikzpicture}
      \end{minipage}
    \end{tabular}
  \end{center}
  \caption{\label{fig:classif} A first classification of dynamic networks, based on evolving graph properties that result from the analysis of Section~\ref{sec:analysis}.}
\end{figure}

Further classes were introduced in the recent literature, and organized into a classification in \cite{CFQS10}. They include ${\cal F}_8$ ({\em round connectivity}): every node can reach every other node, and be reached back afterwards; ${\cal F}_9$: ({\em recurrent connectivity}): every node can reach all the others infinitely often; ${\cal F}_{10}$ ({\em recurrence of edges}): the underlying graph $G=(V,E)$ is connected, and every edge in $E$ re-appears infinitely often; ${\cal F}_{11}$ ({\em time-bounded recurrence of edges}): same as ${\cal F}_{10}$, but the re-appearance is bounded by a given time duration; ${\cal F}_{12}$ ({\em periodicity}): the underlying graph $G$ is connected and every edge in $E$ re-appears at regular intervals; ${\cal F}_{13}$ ({\em eventual instant-routability}): given any pair of nodes and at any time, there always exists a future $G_i$ in which a (static) path exists between them; ${\cal F}_{14}$ ({\em eventual instant-connectivity}): at any time, there always exists a future $G_i$ that is connected in a classic sense ({\it i.e.,} a static path exists in $G_i$ between any pair of nodes); ${\cal F}_{15}$ ({\em perpetual instant-connectivity}): every $G_i$ is connected in a static sense; ${\cal F}_{16}$ ({\em T-interval-connectivity}): all the graphs in any sub-sequence $G_i, G_{i+1}, ... G_{i+T}$ have at least one connected spanning subgraph in common. Finally, ${\cal F}_{17}$ is the reference class for population protocols, it corresponds to the subclass of ${\cal F}_{10}$ in which the underlying graph $G$ ({\em graph of interaction}) is a complete graph.

All these classes were shown to have particular algorithmic significance. For example, ${\cal F}_{16}$ allows to speed up the execution of some algorithms by a factor $T$~\cite{KLO10}. In a context of broadcast, ${\cal F}_{15}$ allows to have at least one new node informed in every $G_i$, and consequently to bound the broadcast time by (a constant factor of) the network size~\cite{OW05}. ${\cal F}_{13}$ and ${\cal F}_{14}$ were used in~\cite{RBK07} to characterize the contexts in which {\em non-delay-tolerant} routing protocols can eventually work if they retry upon failure. Classes ${\cal F}_{10}$, ${\cal F}_{11}$, and ${\cal F}_{12}$ were shown to have an impact on the distributed versions of {\em foremost}, {\em shortest}, and {\em fastest} broadcasts with termination detection. Precisely, foremost broadcast is feasible in ${\cal F}_{10}$, whereas shortest and fastest broadcasts are not; shortest broadcast becomes feasible in ${\cal F}_{11}$~\cite{CFMS10}, whereas fastest broadcast is not and becomes feasible in ${\cal F}_{12}$. Also, even though foremost broadcast is possible in ${\cal F}_{10}$, the memorization of the journeys for subsequent use is not possible in ${\cal F}_{10}$ nor ${\cal F}_{11}$; it is however possible in ${\cal F}_{12}$~\cite{CFMS11}. Finally, ${\cal F}_8$ could be regarded as a {\em sine qua non} for termination detection in many contexts.

Interestingly, this new range of classes --~from ${\cal F}_{8}$ to ${\cal F}_{17}$~-- can also be integrally connected by means of a set of inclusion relations, as illustrated on Figure~\ref{fig:classif2}. Both classifications can also be inter-connected through ${\cal F}_{8}$, a subclass of ${\cal F}_2$, which brings us to $17$ connected classes. A classification of this type can be useful in several respects, including the possibility to transpose results or to compare solutions or problems on a formal basis, which we discuss now.

\begin{figure}[h]
  \centering
\begin{tikzpicture}[sloped, level distance=40pt, font=\small, sibling distance=30pt]
  \tikzstyle{every path}=[<-]
  \node {${\cal F}_8$}
  child [grow=right]{
    node {${\cal F}_2$}
    edge from parent [->, dashed]
  }
  child [grow=left]{
    node {${\cal F}_9$}
    child {
      node {${\cal F}_{10}$}
      child {
        node {${\cal F}_{17}$}
      }
      child [sibling distance=16pt] {
        node {${\cal F}_{11}$}
        child {
          node {${\cal F}_{12}$}
        }
      }
    }
    child {
      node {${\cal F}_{13}$}
      child {
        node {${\cal F}_{14}$}
        child {
          node {${\cal F}_{15}$}
          child {
            node {${\cal F}_{16}$}
          }
        }
      }
    }
  };
\end{tikzpicture}
\caption{\label{fig:classif2} Complementary classification, based on further classes found in the recent literature (figure from~\cite{CFQS11}).}
\end{figure}
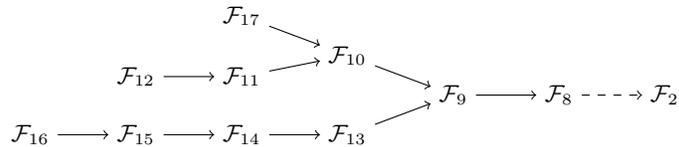

\subsection{Comparison of algorithms based on their topological requirements} 

Let us consider the two counting algorithms given in Section~4. To have any chance of success, $\mathcal{A}_2$ requires the evolving graph to be in $\mathcal{F}_5$ (with a fortunate choice of counter) or in $\mathcal{F}_6$ (with any vertex as counter). On the other hand, $\mathcal{A}_3$ requires the evolving graph to be in $\mathcal{F}_7$. Since both $\mathcal{F}_5$ (directly) and $\mathcal{F}_6$ (transitively) are included in $\mathcal{F}_7$, there are some topological scenarios (\textit{i.e.,} $\mathcal{G} \in {\mathcal{F}_7}{\setminus \mathcal{F}_5}$) in which $\mathcal{A}_2$ has no chance of success, while $\mathcal{A}_3$ has some. Such observation allows to claim that $\mathcal{A}_3$ is more general than $\mathcal{A}_2$ with respect to its topological requirements. This illustrates how a classification can help compare two solutions on a fair and formal basis. In the particular case of these two counting algorithms, however, the claim could be balanced by the fact that a sufficient condition is known for $\mathcal{A}_2$, whereas none is known for $\mathcal{A}_3$. The choice for the right algorithm may thus depend on the target mobility context: if this context is thought to produce topological scenarios in $\mathcal{F}_5$ or $\mathcal{F}_6$, then $\mathcal{A}_2$ could be preferred, otherwise $\mathcal{A}_3$ should be considered. 

A similar type of reasoning could also teach us something about the problems themselves. Consider the above-mentioned results about {\em shortest}, {\em fastest}, and {\em foremost} broadcast with termination detection, the fact that ${\cal F}_{12}$ is included in ${\cal F}_{11}$, which is itself included in ${\cal F}_{10}$, tells us that there is a (at least partial) order between these problems topological requirements: $foremost \preceq shortest \preceq fastest$. 

We believe that classifications of this type have the potential to lead more equivalence results and formal comparison between problems and algorithms. Now, one must also keep in mind that these are only \emph{topology-related} conditions, and that other dimensions of properties --~{\it e.g.,} what knowledge is available to the nodes, or whether they have unique identifiers~-- keep playing the same important role as they do in a static context. Considering again the same example, the above classification hides that detecting termination in the {\em foremost} case in ${\cal F}_{10}$ requires the emitter to know the number of nodes $n$ in the network, whereas this knowledge is not necessary for shortest broadcast in ${\cal F}_{11}$ (the alternative knowledge of knowing a bound on the recurrence time is sufficient). In other words, lower topology-related requirements do not necessarily imply lower requirements in general.

\section{Mechanization potential}
\label{sec:mechanization}

One of the motivations of this work is to contribute to the development of assistance tools for algorithmic design and decision support in mobile ad hoc networks. The usual approach to assess the correct behavior of an algorithm or its appropriateness to a particular mobility context is to perform simulations. A typical simulation scenario consists in executing the algorithm concurrently with topological changes that are generated using a {\em mobility model} ({\it e.g.,} the {\em random way point} model, in which every node repeatedly selects a new destination at random and moves towards it), or on top of real network traces that are first collected from the real world, then replayed at simulation time. As discussed in the introduction, the simulation approach has some limitations, among which generating results that are difficult to generalize, reproduce, or compare with one another on a non-subjective basis.

The framework presented in this paper allows for an analytical alternative to simulations. The previous section already discussed how two algorithms could be compared on the basis of their topological requirements. We could actually envision a larger-purpose chain of operations, aiming to characterize how appropriate a given algorithm is to a given mobility context. The complete workflow is depicted on Figure~\ref{fig:checking}. 

\begin{figure}[h]
  \centering
  \begin{tikzpicture}[xscale=.81, yscale=1.2]
    \tikzstyle{every node}=[font=\scriptsize];
    \path[right] (0,2) node (algo) {Algorithm};
    \path[right] (0,1.3) node (mobi) {Mobility Model};
    \path[right] (0,0.7) node (real) {Real Network};
    \path[right] (7.7,2) node (clas) {Dynamic Graph Classes};
    \path[right] (7.7,1) node (inst) {Dynamic Graph Instances};
    \path[right] (4.8,2) node (cond) {Conditions};
    \path[right] (4.8,1.3) node (tra1) {Network Traces};
    \path[right] (4.8,0.7) node (tra2) {Network Traces};
    \tikzstyle{every node}=[draw, rectangle, rounded corners=6pt, font=\footnotesize];
    \path (3.65,2) node (anal) {Analysis};
    \path (3.65,1.3) node (gene) {Generation};
    \path (3.65,0.7) node (sens) {Collection};
    \path[right] (11.8,1.5) node (chec) {Inclusion Checking};
    \tikzstyle{every node}=[];
    \path[left] (chec)+(0,-.9) node (yes) {Yes};
    \path[right] (chec)+(0,-.9) node (no) {No};
    \path (chec.west)+(0,.1) coordinate (c1);
    \path (chec.west)+(0,-.1) coordinate (c2);
    \path (inst.west)+(0,.08) coordinate (i1);
    \path (inst.west)+(0,-.08) coordinate (i2);
    \tikzstyle{every path}=[->, draw, thin];
    \draw (algo)--(anal);
    \draw (anal)--(cond);
    \draw (cond)--(clas);
    \draw (mobi)--(gene);
    \draw (gene)--(tra1.west);
    \draw (real)--(sens);
    \draw (sens)--(tra2.west);
    \draw (chec)--(yes);
    \draw (chec)--(no);
    \tikzstyle{every path}=[->, draw, thin, shorten >=1.5pt];
    \draw (clas.east)--(c1);
    \draw (inst.east)--(c2);
    \draw (tra1.east)--(i1);
    \draw (tra2.east)--(i2);
  \end{tikzpicture}
  \caption{\label{fig:checking} Automated checking of the suitability of an algorithm in various mobility contexts.}
\end{figure}
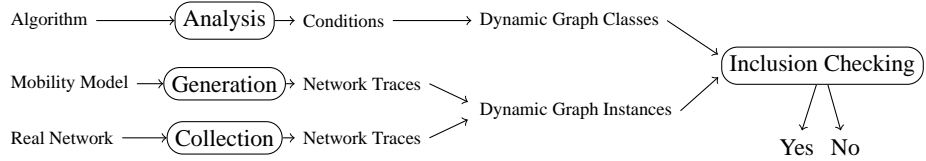

On the one hand, algorithms are analyzed, and necessary/sufficient conditions determined. This step produces {\em classes} of evolving graphs. On the other hand, mobility models and real-world networks can be used to generate a collection of network traces, each of which corresponds to an {\em instance} of evolving graphs. Checking how given instances distribute within given classes --~{\it i.e.,} are they included or not, in what proportion?~-- may give a clue about the appropriateness of an algorithm in a given mobility context. This section starts discussing the question of understanding to what extent such a workflow could be automated ({\em mechanized}), in particular through the two core operations of {\em Inclusion checking} and {\em Analysis}, both capable of raising problems of a theoretical nature.

\subsection{Checking network traces for inclusion in the classes}
We provide below an efficient solution to check the inclusion of an evolving graph in any of the seven classes of Figure~\ref{fig:classif} --~that are, all classes derived from the analysis carried out in Section~\ref{sec:analysis}. Interestingly, each of these classes allows for efficient checking strategies, provided a few transformations are done. The \emph{transitive closure} of the journeys of an evolving graph $\mathcal{G}$ is the graph $H = (V,A_{H})$, where $A_{H} = \{(v_i,v_j) : v_i \leadsto v_j)\}$. Because journeys are oriented entities, their transitive closure is by nature a \emph{directed} graph (see Figure~\ref{fig:closure}). As explained in~\cite{BF03}, the computation of transitive closures can be done efficiently, in $O(|V|.|E|.(log|\ST|.log|V|)$ time, by building the tree of \emph{shortest} journeys from each node in the network. 
We extend this notion to the case of strict journeys, with $H_{strict} = (V,A_{H_{strict}})$, where $A_{H_{strict}} = \{(v_i,v_j) : v_i \leadstoo v_j)\}$. 

\begin{figure}[h]
\def\wgraph (#1,#2){%
  \tikzstyle{every node}=[draw,fill,circle,inner sep=#1]
  \path (0,0) node (c){};
  \path (c)+(-1,.4) node (a){};
  \path (c)+(1,.4) node (e){};
  \path (c)+(-.6,-.6) node (b){};
  \path (c)+(.6,-.6) node (d){};
  \tikzstyle{every node}=[font=\footnotesize]
  \path (a)+(-#2,0) node (la){$a$};
  \path (b)+(-#2,0) node (lb){$b$};
  \path[above] (c) node (lc){$c$};
  \path (d)+(#2,0) node (ld){$d$};
  \path (e)+(#2,0) node (le){$e$};
}
\centering
  \begin{tikzpicture}[scale=2]
    \wgraph(.8pt,4pt)
    \tikzstyle{every node}=[sloped,below,font=\scriptsize,inner sep=2pt]
    \draw (a)--node{$[1,2)$}(b);
    \draw (b)--node{$[2,3)$}(d);
    \draw (d)--node{$[3,5)$}(e);
    \tikzstyle{every node}=[sloped,above,font=\scriptsize,inner sep=1pt]
    \draw (a)--node[inner sep=1.5pt]{$[1,2)$}(c);
    \draw (b)--node[inner sep=1.5pt]{$[2,3)$}(c);
    \draw (c)--node{$[2,4)$}(d);
    \draw (c)--node{$[3,5)$}(e);
    
    \tikzstyle{every node}=[draw, fill, inner sep=.1pt]
    \path (2.6,-.1) coordinate (zz);
    \path (zz)+(18:.6) node (ee) {};
    \path (zz)+(90:.6) node (cc) {};
    \path (zz)+(162:.6) node (aa) {};
    \path (zz)+(234:.6) node (bb) {};
    \path (zz)+(306:.6) node (dd) {};
    \tikzstyle{every node}=[]
    \path (aa)+(-.1,.02) node (laa){$a$};
    \path (bb)+(-.08,-.08) node (laa){$b$};
    \path (cc)+(0,.09) node (laa){$c$};
    \path (dd)+(.08,-.08) node (laa){$d$};
    \path (ee)+(.1,.02) node (laa){$e$};
    
    \tikzstyle{every path}=[>=stealth, shorten >=1.5pt, shorten <=1.5pt]
    \path[<->] (aa) edge [bend right=30] (bb);
    \path[<->] (aa) edge [bend left=30] (cc);
    \path[->] (aa) edge [bend right=10] (dd);
    \path[->] (aa) edge [bend left=10] (ee);
    \path[<->] (bb) edge [bend left=10] (cc);
    \path[<->] (bb) edge [bend right=30] (dd);
    \path[->] (bb) edge [bend right=10] (ee);
    \path[<->] (cc) edge [bend left=10] (dd);
    \path[<->] (cc) edge [bend left=30] (ee);
    \path[<->] (dd) edge [bend right=30] (ee);
    
    \path (zz)+(-1.4,0) coordinate (zzleft);
    \draw[->, semithick, shorten >=50pt] (zzleft)--(zz);
  \end{tikzpicture}
\caption{\label{fig:closure}Example of transitive closure of the journeys of an evolving graph.}
\end{figure}
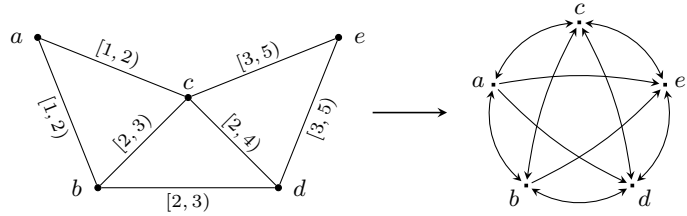

Given an evolving graph $\mathcal{G}$, its underlying graph $G$, its transitive closure $H$, and the transitive closure of its \emph{strict} journeys $H_{strict}$, the inclusion in each of the seven classes can be tested as follows:
\begin{itemize}
\item $\mathcal{G} \in \mathcal{F}_1 \Longleftrightarrow$ $H$ contains an out-dominating set of size 1.
\item $\mathcal{G} \in \mathcal{F}_2 \Longleftrightarrow$ $H$ is a complete graph.
\item $\mathcal{G} \in \mathcal{F}_3 \Longleftrightarrow$ $H_{strict}$ contains an out-dominating set of size 1.
\item $\mathcal{G} \in \mathcal{F}_4 \Longleftrightarrow$ $H_{strict}$ is a complete graph.
\item $\mathcal{G} \in \mathcal{F}_5 \Longleftrightarrow$ $G$ contains a dominating set of size 1.
\item $\mathcal{G} \in \mathcal{F}_6 \Longleftrightarrow$ $G$ is a complete graph.
\item $\mathcal{G} \in \mathcal{F}_7 \Longleftrightarrow$ $H$ contains an in-dominating set of size 1.
\end{itemize}

How the classes of Figure~\ref{fig:classif2} could be checked is left open. Their case is more complex, or at least substantially different, because the corresponding definitions rely on the notion of {\em infinite}, which a network trace is necessarily not. For example, whether a given edge is eventually going to reappear (e.g. in the context of checking inclusion to class $\mathcal{F}_8$ or $\mathcal{F}_9$) cannot be inferred from a finite sequence of events. However, it is certainly feasible to check whether a {\em given} recurrence bound applies within the time-span of a {\em given} network trace (bounded recurrence $\mathcal{F}_{10}$), or similarly, whether the sequence of events repeats modulo $p$ (for a given $p$) within the given trace (periodic networks $\mathcal{F}_{11}$). 

\subsection{Towards a mechanized analysis}

The most challenging component of the workflow on Figure~\ref{fig:checking} is certainly that of {\em Analysis}. Ultimately, one may hope to build a component like that of Figure~\ref{fig:analysis}, which is capable of answering whether a given property is necessary (no possible success without), sufficient (no possible failure with), or orthogonal (both success and failure possible) to a given algorithm with given computation assumptions ({\it e.g.,} a particular type of synchronization or progression hypothesis). Such a workflow could ultimately be used to confirm an intuition of the analyst, as well as to discover new conditions automatically, based on a collection of properties.

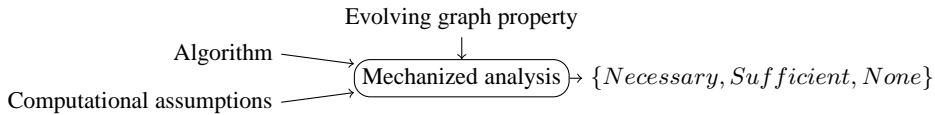
\begin{figure}[h]
  \centering
  \begin{tikzpicture}[scale=.8]
    \tikzstyle{every node}=[font=\footnotesize];
    \path (0,.4) node[left] (algo) {Algorithm};
    \path (3,1) node (prop) {Evolving graph property};
    \path (0,-.4) node[left] (assu) {Computational assumptions};
    \path (5,0) node[right] (resu) {$\{Necessary, Sufficient, None\}$};
    \tikzstyle{every node}=[draw, rectangle, rounded corners=6pt, font=\footnotesize];
    \path (3,0) node (anal) {Mechanized analysis};
    \tikzstyle{every path}=[->];
    \draw (algo.east)--(anal);
    \draw (prop)--(anal);
    \draw (assu.east)--(anal);
    \draw (anal)--(resu);
  \end{tikzpicture}
  \caption{\label{fig:analysis}Possible interface for a mechanized analysis.}
\end{figure}

As of today, such an objective is still far from reach, and a number of intermediate steps should be taken. For example, one may consider specific {\em instances} of evolving graphs rather than general properties. We develop below a prospective idea inspired by the work of Cast{\'e}ran {\it et al.} in static networks~\cite{coq09,coq11}. Their work focus on bridging the gap between local computations and the formal proof management system {\em Coq}~\cite{coq-manual}, and materializes, among others, as the development of a {\em Coq} library: {\em Loco}. This library contains appropriate representations for graphs and labelings in {\em Coq} (by means of {\em sets} and {\em maps}), as well as an operational description of relabeling rule execution (see Section~6 of~\cite{coq11} for details). The fact that such a machinery is already developed is worthwhile noting, because we believe {\em evolving graphs} could be seen themselves as relabelings acting on a 'presence' label on vertices and edges. The idea in this case would be to re-define topological events as being themselves graph relabeling rules whose {\em preconditions} correspond to a $G_i$ and {\em actions} lead to the next $G_{i+1}$. Considering the execution of these rules concurrently with those of the studied algorithm could make it possible to leverage the power of {\em Coq} to mechanize proofs of correctness and/or impossibility results in given instances of evolving graphs. 

\section{Concluding remarks and open problems}
\label{sec:conclusion}

This paper suggested the combination of existing tools and the use of dedicated methods for the analysis of distributed algorithms in dynamic networks. The resulting framework allows to characterize assumptions that a given algorithm requires in terms of topological evolution during its execution. We illustrated it by the analysis of three basic algorithms, whose necessary and sufficient conditions were derived into a sketch of classification of dynamic networks. We showed how such a classification could be used in turn to compare algorithms on a formal basis and provide assistance in the selection of an algorithm. This classification was extended by an additional $10$ classes from recent literature. We finally discussed some implications of this work for mechanization of both decision support systems and analysis, including respectively the question of checking whether a given network trace belongs to one of the introduced classes, and prospective ideas on the combination of evolving graph and graph relabeling systems within the {\em Coq} proof assistant.

Analyzing the network requirements of algorithms is not a novel approach in general. It appears however that it was never considered in systematic manner for {\em dynamics}-related assumptions. Instead, the apparent norm in dynamic network analytical research is to study problems {\em once} a given set of assumptions has been considered, these assumptions being likely chosen for analytical convenience. This appears particularly striking in the recent field of \emph{population protocols}, where a common assumption is that a pair of nodes interacting once will interact infinitely often. In the light of the classification shown is this paper, such an assumption corresponds to a highly specific computing context. We believe the framework in this paper may help characterize weaker topological assumptions for the same class of problems.

Our work being mostly of a conceptual essence, a number of questions may be raised relative to its broader applicability. For example, the algorithms studied here are simple. A natural question is whether the framework will scale to more complex algorithms. We hope it could suit the analysis of most fundamental problems in distributed computing, such as \emph{election}, \emph{naming}, \emph{concensus}, or the construction of \emph{spanning structures} (note that \emph{election} and \emph{naming} may not have identical assumptions in a dynamic context, although they do in a static one). Our discussion on mechanization potentials left two significant questions undiscussed: how to check for the inclusion of an evolving graph in all the remaining classes, and how to approach the problem of mechanizing analysis relative to a general property. Another prospect is to investigate how intermediate properties could be explored between necessary and sufficient conditions, for example to guarantee a desired {\em probability} of success. Finally, besides these characterizations on feasibility, one may also want to look at the impact that particular properties may have on the {\em complexity} of problems and algorithms. Analytical research in dynamic networks is still in its infancy, and many exiting questions remain to be explored.


\section{Acknowledgments}
We are grateful to Pierre Cast\'eran and Vincent Filou for bringing our attention to the possible connections between this work and formal proof systems.
\newcommand{\etalchar}[1]{$^{#1}$}

\end{document}